\documentclass[11pt,a4paper,final]{amsart}
\usepackage{amssymb,latexsym}
\usepackage{enumerate}
\usepackage{pifont,wasysym}
\usepackage[T1]{fontenc}
\usepackage{color}

\usepackage[left=3cm,right=3cm,top=3cm,bottom=3cm]{geometry} 
\usepackage{hyperref} 

\theoremstyle{plain}
\newtheorem{theorem}{Theorem}
\newtheorem{corollary}{Corollary}[section]
\newtheorem{lemma}[corollary]{Lemma}
\newtheorem{proposition}[corollary]{Proposition}

\theoremstyle{remark}
\newtheorem{remark}[corollary]{Remark}

\newcommand{\Funkcyjne}[1]{#1}

\newcommand{\NN}{\mathbb{N}}
\newcommand{\ZZ}{\mathbb{Z}}
\newcommand{\QQ}{\mathbb{Q}}
\newcommand{\RR}{\mathbb{R}}
\newcommand{\CC}{\mathbb{C} }
\newcommand{\TT}{\mathbb{T}} 
\newcommand{\II}{\mathsf{i}}

\newcommand{\ProjectorExpanded}[1]{\left|#1 \right\rangle \left\langle #1\right|}
\newcommand{\SimpleOperatorExpanded}[2]{\left|#1 \right\rangle \left\langle #2\right|}

\newcommand{\Angle}[3]{\angle\ScalarProduct{#1}{#2}{#3}}
\newcommand{\Mod}{\operatorname{mod}}
\newcommand{\CConjugate}[1]{\overline{#1}}
\newcommand{\ComplexConjugate}[1]{\CConjugate{#1}}
\newcommand{\TruncatedNumbers}[1]{= 1,\ldots,#1}
\newcommand{\Set}[1]{\left\{ #1  \right\}}
\newcommand{\SetSuchThat}[2]{\left\{\, #1 \ | \ #2 \, \right\}}
\newcommand{\Supp}{\operatorname{supp}}
\newcommand{\Cl}{\operatorname{cl}}
\newcommand{\IndicatorFunction}[1]{\mathrm{\mathbf{1}}_{#1}}
\newcommand{\Cardinality}[1]{\# #1}
\newcommand{\HilbertSpace}[1]{\mathcal{H}^{#1}}
\newcommand{\HProduct}[3]{\left\langle #1 \mid #2 \right\rangle_{#3}}
\newcommand{\HermitianProduct}[3]{\HProduct{#1}{#2}{#3}}
\newcommand{\ScalarProduct}[3]{(#1\mid #2)_{#3}}
\newcommand{\HermitianNorm}[2]{\left|#1\right|_{#2}}
\newcommand{\Conjugate}[1]{{#1}^*}
\newcommand{\Bra}[1]{\left\langle #1\right|}
\newcommand{\Ket}[1]{\left|#1\right\rangle}
\newcommand{\Multiplicity}{m}
\newcommand{\HilbertTensorProduct}[1]{\HilbertSpace{1}\otimes\dots\otimes\HilbertSpace{#1}}
\newcommand{\TensorProduct}[2]{#1^{1}\otimes\dots\otimes #1^{#2}}
\newcommand{\BoundedOperators}[1]{\mathcal{B}(#1)}
\newcommand{\HermitianOperators}[1]{\mathcal{A}(#1)} 
\newcommand{\HermitianOperatorsDual}[1]{\HermitianOperators{#1}^*}
\newcommand{\Operator}{\rho}
\newcommand{\Projector}[1]{\mathcal{P}\left(#1\right)}
\newcommand{\PositiveOperators}[1]{\mathcal{P}(#1)}
\newcommand{\SeparableOperators}[1]{\mathcal{S}(#1)}
\newcommand{\DensityOperators}[1]{\mathcal{D}(#1)}
\newcommand{\Cone}{\mathrm{C}}
\newcommand{\DualCone}[1]{{#1}^*}
\newcommand{\FourierTransform}[1]{\mathcal{F}\left(#1\right)}
\newcommand{\InverseFourierTransform}[1]{\mathcal{F}^{-1}\left(#1\right)}
\newcommand{\FourierTransformMapping}{\mathcal{F}}
\newcommand{\InverseFourierTransformMapping}{\mathcal{F}^{-1}}
\newcommand{\Group}[1]{\mathsf{#1}}
\newcommand{\NeutralElementDual}{\Group{\epsilon}}
\newcommand{\Measure}[1]{d #1}
\newcommand{\DualGroup}[1]{{\Group{#1}}^*}
\newcommand{\DiscreteGroup}{\Group{D}}
\newcommand{\Lattice}{\mathrm{\Lambda}}
\newcommand{\QuotientGroup}[2]{\mathit{#1 / #2}}
\newcommand{\IntegrableFunctions}[2]{\LL{1}(#1,#2)}
\newcommand{\Function}{\phi}
\newcommand{\IntegralMappingExpanded}[3]{\int_{#1} #2 \, #3}
\newcommand{\IdotsIntegralMappingExpanded}[3]{\underset{#1}\idotsint #2 \, #3}
\newcommand{\LL}[1]{\Funkcyjne{L}_{#1}}
\newcommand{\Lnorm}[2]{\left\|#1\right\|_{\Funkcyjne{L}_{#2}}}
\newcommand{\SquareIntegrableFunctions}[2]{\Funkcyjne{L}_2(#1,#2)}
\newcommand{\OperatorRepresentedBy}[1]{\Operator(#1)}
\newcommand{\DualOperator}{\omega}

\newcommand{\Character}{\gamma}
\newcommand{\Action}[1]{\chi_{#1}}
\newcommand{\Convolution}[2]{#1\star#2} 
\newcommand{\MultipleConvolution}[2]{\Convolution{#1^1}{\Convolution{\ldots}{#1^{#2}}}}
\newcommand{\TensorConvolution}[2]{#1\logof#2} 
\newcommand{\MultipleTensorConvolution}[2]{\left(\TensorConvolution{#1^1}{\TensorConvolution{\ldots}{#1^{#2}}}\right)}
\newcommand{\ReferenceSetConvolutional}{\prod_{\mu} \left( \IntegrableFunctions{\Group{G}}{\HilbertSpace{\mu}} \cap \SquareIntegrableFunctions{\Group{G}}{\HilbertSpace{\mu}} \right)}
\newcommand{\LCA}{LCA}
\newcommand{\SOHS}{\SeparableOperators{\HilbertSpace{}}}
\newcommand{\EE}[1]{e^{#1}}
\newcommand{\Span}{\operatorname{span}}
\newcommand{\Identity}[1]{\operatorname{Id_{#1}}}

\begin{document}

\title{On the tensor convolution and the quantum separability problem}

\author[G. Pietrzkowski]{Gabriel Pietrzkowski}
\address{Institute of Mathematics, Polish Academy of Sciences\\
00-956 Warszawa, Poland}
\email{G.Pietrzkowski@impan.pl}

\date{}

\begin{abstract}
We consider the problem of separability: decide whether a Hermitian operator on a finite dimensional Hilbert tensor product $\HilbertSpace{}=\HilbertTensorProduct{\Multiplicity}$ is separable or entangled. We show that the tensor convolution $\MultipleTensorConvolution{\Function}{\Multiplicity}:\Group{G}\to\HilbertSpace{}$ defined for mappings $\Function^\mu:\Group{G}\to\HilbertSpace{\mu}$ on an almost arbitrary locally compact abelian group $\Group{G}$, give rise to formulation of an equivalent problem to the separability one.
\end{abstract}



\maketitle

\section{Introduction}

The problem of separability for a given Hilbert tensor product $\HilbertSpace{} = \HilbertTensorProduct{\Multiplicity}$ is to determine whether a given positive semi-definite operator $\Operator\in\PositiveOperators{\HilbertTensorProduct{\Multiplicity}}$ is separable or not, i.e. if it can be written as a finite sum 
\begin{align}
\label{eqn:SeprableOperator}
\Operator = \sum_p \lambda_p\ProjectorExpanded{\TensorProduct{v_p}{\Multiplicity}},
\end{align}
where $\lambda_p > 0$ and $v_p^\mu\in\HilbertSpace{\mu}$. Note that in quantum mechanics the density operators $\DensityOperators{\HilbertTensorProduct{\Multiplicity}}$ are considered instead of positive semi-definite ones, and then the separability problem is to recognize if $\Operator\in\DensityOperators{\HilbertTensorProduct{\Multiplicity}}$ is of the form (\ref{eqn:SeprableOperator}) where additionally $\sum_p\lambda_p=1$, and $v_p^\mu\in\HilbertSpace{\mu}$ have norm 1. Clearly, both problems are equivalent -- we consider the first one only for technical reasons.

The main obstacle to efficient solution of this problem is that each separable mixed state (except projectors on simple tensors in a Hilbert tensor product -- separable projectors) is ambiguously decomposable into a convex combination of separable projectors \eqref{eqn:SeprableOperator}. Therefore, if we take, for example, a Hilbert space $\HilbertSpace{} = \HilbertSpace{1}\otimes\HilbertSpace{2}$ and arbitrary $v^1,w^1\in\HilbertSpace{1}$, $v^2,w^2\in\HilbertSpace{2}$, then it is not so simple to recognize that
\begin{equation}
\label{eqn:One}
\ProjectorExpanded{v_0\otimes w_0 + v_1\otimes w_1} + \ProjectorExpanded{v_0\otimes w_1 + v_1\otimes w_0}
\end{equation}
is always a separable operator (of course one can use the Horodeckis' criterion \cite{horodecki} in order to check this). But even if we know that it is separable, it is not immediate to see that it can be written in the form
\begin{equation}
\label{eqn:Two}
\tfrac{1}{2} \Projector{(v_0+v_1)\otimes(w_0+w_1)} + \tfrac{1}{2} \Projector{(v_0-v_1)\otimes(w_0-w_1)},
\end{equation}
where $\Projector{v} = \ProjectorExpanded{v}$ for $v\in\HilbertSpace{}$.
In this article we propose a method of transforming operators on an $\Multiplicity$-fold Hilbert tensor product written in a special form -- in analogy with (\ref{eqn:One}) -- into the positive combination of separable projectors -- like in (\ref{eqn:Two}). 
Moreover, we prove that each separable operator and none of entangled ones can be written in such form (see Theorem \ref{thm:3}). In other words, we restate the problem of separability (see Remark \ref{rem:Separability} after Theorem \ref{thm:3}). 

Entanglement has been considered since 1935 when the EPR paradox was formulated \cite{einstein} by the famous triple Einstein, Podolsky and Rosen. Nevertheless the history of the separability problem begun when Werner defined the notion of separability of bipartite mixed states in his celebrated article \cite{werner}. Years of research convinced scientists that the problem is not so easy to solve. Horodeckis' showed connection of the problem of serability for bipartite states with the one of classifying the so called positive maps \cite{horodecki}, and then generalized this result to the multipartite case \cite{horodecki2}. They also formulated an efficient criterion for the problem in case 
$\HilbertSpace{1}\otimes\HilbertSpace{2}$ is such that $\dim\HilbertSpace{1} \cdot \dim\HilbertSpace{2} \leq 6$, when all positive maps can be expressed in terms of completely positive maps and a partial transpose map (in a certain basis) \cite{stromer,woronowicz,peres}. The intuitive assertion that the problem of separability is very hard to solve, has been formally proven by Gurvits who showed NP-hardness of the problem \cite{gurvits} and recently Gharibian strengthen this result \cite{gharibian} proving its strong NP-hardness (see also \cite{ioannou}). A compact group theoretical approach to the problem has been proposed by Korbicz, Wehr and Lewenstein in \cite{korbicz1,korbicz2} and recently they formulated a similar quantum group approach \cite{korbicz3}. Also geometry of the cone of separable operators (or the convex body of seprable mixed states), strictly related to the separability problem, was studied for example by  Ku\' s and \. Zyczkowski \cite{kus}, Grabowski, Ku\' s and Marmo \cite{grabowski}.
In practice the most useful, in low dimensions, is a numerical test discovered by Doherty, Parrilo and Spedalieri \cite{doherty1,doherty2,doherty3}. 

\section{Results}

Let $\HilbertSpace{} = \HilbertTensorProduct{\Multiplicity}$ be a Hilbert tensor product of finite dimension.
Consider an arbitrary locally compact abelian (\LCA) group $\Group{G}$ with the Haar measure $\Measure{g}$ on it.  
A \emph{tensor convolution} $\MultipleTensorConvolution{\Function}{\Multiplicity}:\Group{G}\to\HilbertSpace{}$ of mappings $\Function^\mu:\Group{G}\to\HilbertSpace{\mu}$ is defined inductively as follows:
\begin{align*}
\TensorConvolution{\Function^1}{\Function^2}(g) &=
\IntegralMappingExpanded
{\Group{G}}
{\Function^1(g-h)\otimes\Function^2(h)}
{\Measure{h}}
\in \HilbertSpace{1}\otimes\HilbertSpace{2},
\\
\vdots
\\
\MultipleTensorConvolution{\Function}{\Multiplicity}(g) &=  
\IntegralMappingExpanded
{\Group{G}}
{\MultipleTensorConvolution{\Function}{\Multiplicity-1}(g-h) \otimes\Function^\Multiplicity(h)}
{\Measure{h}}
\in \HilbertSpace{}.
\end{align*}
Equivalently, we can write
\begin{multline}
\label{eqn:TensorConvolution}
\MultipleTensorConvolution{\Function}{\Multiplicity}(g) = \\
 \IdotsIntegralMappingExpanded
{\Group{G}^{m-1}}
{\Function^1(g - g_2-\ldots-g_m)\otimes \Function^2(g_2) \otimes\cdots\otimes \Function^m(g_m)}
{\Measure{g_2}\cdots\Measure{g_m}}.
\end{multline}

The concept of the tensor convolution is based on that of usual convolution of functions. However, it is highly asymmetric. For example, assume that for $\mu\TruncatedNumbers{\Multiplicity}$ there exist $f^\mu:\Group{G}\to\CC$ and a vector $v^\mu\in\HilbertSpace{\mu}$ such that $\Function^\mu = f^\mu\cdot v^\mu$. Then the tensor convolution of $\Function^\mu$'s equals the convolution of $f^\mu$'s multiplied by the tensor product of $v^\mu$'s, that is
$$
\MultipleTensorConvolution{\Function}{\Multiplicity} = \MultipleConvolution{f}{\Multiplicity}\cdot\TensorProduct{v}{\Multiplicity},
$$
where $\Convolution{}{}$ denotes the standard convolution of functions on the group.

In section \ref{sec:Proof1} we prove the following theorem.

\begin{theorem}
\label{thm:Convolutional}
Let $\Group{G}$ be a locally compact abelian group with the Haar measure $\Measure{g}$ on it. Assume that $\Function^\mu:\Group{G}\to\HilbertSpace{\mu}$ is absolutely and square integrable, for $\mu\TruncatedNumbers{\Multiplicity}$.  Then the Hermitian operator 
\begin{equation}
\label{eqn:CRepresentation}
\tag{\ding{73}}
\OperatorRepresentedBy{\Function} = \IntegralMappingExpanded
{\Group{G}} 
{\ProjectorExpanded{\MultipleTensorConvolution{\Function}{\Multiplicity}(g)}} 
{\Measure{g}},
\end{equation}
acting on $\HilbertTensorProduct{\Multiplicity}$, is separable.
\end{theorem}

Let us consider the example from the introduction once again.
Take $\Group{G} = \ZZ_2$ with the counting measure as the Haar measure. Consider a Hilbert space $\HilbertSpace{} = \HilbertSpace{1}\otimes\HilbertSpace{2}$. Let $\Function^1(g) = v_g\in\HilbertSpace{1}$ and $\Function^2(g) = w_g\in\HilbertSpace{2}$ for $g\in\ZZ_2$. Then the tensor convolution $\TensorConvolution{\Function^1}{\Function^2}$ takes the values
\begin{align*}
\TensorConvolution{\Function^1}{\Function^2}(0) &= v_0\otimes w_0 + v_1\otimes w_1, \\
\TensorConvolution{\Function^1}{\Function^2}(1) &= v_0\otimes w_1 + v_1\otimes w_0.
\end{align*}
The theorem asserts that 
\begin{equation}
\label{eqn:Three}
\ProjectorExpanded{v_0\otimes w_0 + v_1\otimes w_1} + \ProjectorExpanded{v_0\otimes w_1 + v_1\otimes w_0}
\end{equation}
is a separable operator, that is what we pointed out in (\ref{eqn:One}). To see that it equals (\ref{eqn:Two}), we need to recall that for each \LCA~group $\Group{G}$ there exists the dual locally compact abelian group $\DualGroup{G}$. Each element $\Character\in\DualGroup{G}$, which is called a character, represents the continuous function $\Action{\Character}:\Group{G}\to\CC$ (in fact it is a continuous homomorphism of the groups $\Group{G}$ and the unit circle $\TT$ in $\CC$ -- see section \ref{sec:Fourier}). Then for an absolutely integrable functions $f:\Group{G}\to\CC$ we define its Fourier transform $\FourierTransform{f}:\DualGroup{G}\to\CC$ by 
\begin{equation*}
\FourierTransform{f}(\Character) = \IntegralMappingExpanded
{\Group{G}}
{\Action{\Character}(-g)\cdot f(g)}
{\Measure{g}}.
\end{equation*} 
Actually, by the same formula we can define the Fourier transform for an absolutely integrable mapping $\Function^\mu:\Group{G}\to\HilbertSpace{\mu}$  as well. Finally, since $\DualGroup{G}$ is an \LCA~group it possesses the Haar measure $\Measure{\Character}$, which additionally can by normalized so that the Parseval equality (\ref{eqn:ParsevalEquality}) holds. We say that $\Measure{\Character}$ is conjugated with $\Measure{g}$.  
The main part of the proof of Theorem \ref{thm:Convolutional} is the following proposition (proved in section \ref{sec:Proof1}). 

\begin{proposition}
\label{prop:21}
Under the assumptions of Theorem \ref{thm:Convolutional}, let $\DualGroup{G}$ be the dual group of $\Group{G}$ and $\Measure{\Character}$ the Haar measure conjugated with $\Measure{g}$. Then  
\begin{eqnarray}
\label{eqn:36}
\OperatorRepresentedBy{\Function} = 
\IntegralMappingExpanded
{\DualGroup{G}} 
{\Projector{\FourierTransform{\Function^1}(\Character)\otimes \cdots\otimes\FourierTransform{\Function^{\Multiplicity}}(\Character)}} 
{\Measure{\Character}},
\end{eqnarray}
where $\OperatorRepresentedBy{\Function}\in\SOHS$ is given by (\ref{eqn:CRepresentation}), and $\Projector{v}=\ProjectorExpanded{v}$ for $v\in\HilbertSpace{}$.
\end{proposition}

Coming back to the example, for $\Group{G}=\ZZ_2$ the dual group is $\DualGroup{G} = \ZZ_2$, and a character $\Character\in\DualGroup{G}$ acts by $\Action{\Character}(g) = e^{\II\cdot g\Character}$. Therefore,
\begin{align*}
\FourierTransform{\Function^1}(\Character) &= v_0 + (-1)^\Character\cdot v_1,&
\FourierTransform{\Function^2}(\Character) &= w_0 + (-1)^\Character\cdot w_1.
\end{align*}
Since the counting measure divided by 2 (on $\DualGroup{G}$) is the Haar measure conjugated with the counting measure on $\Group{G}$, we get from the above proposition that (\ref{eqn:Three}) equals
\begin{equation*}
\tfrac{1}{2} \Projector{(v_0+v_1)\otimes(w_0+w_1)} + \tfrac{1}{2} \Projector{(v_0-v_1)\otimes(w_0-w_1)}
\end{equation*}
as we mentioned in the introduction.

The above example is representative in the sense that for a random choice of functions $\Function^1,\ldots,\Function^\Multiplicity$, on an arbitrary \LCA~group, their tensor convolution takes entangled values almost everywhere. Therefore, we represent separable operators by "sums" of entangled projectors.

Having got Theorem \ref{thm:Convolutional}, a natural question occurs: do for each separable operator $\Operator\in\SOHS$ there exist mappings $\Function^\mu:\Group{G}\to\HilbertSpace{\mu}$ such that $\Operator=\OperatorRepresentedBy{\Function}$ given by (\ref{eqn:CRepresentation})? The answer is affirmative provided that the cardinality of the group $\Group{G}$ is at least $(\dim\HilbertSpace{})^2$.

\begin{theorem}
\label{thm:3}
Let $\HilbertSpace{}=\HilbertTensorProduct{\Multiplicity}$ be an m-fold tensor product of finite dimensional Hilbert spaces $\HilbertSpace{\mu}$.
Let $\Group{G}$ be a locally compact abelian group of cardinality $\Cardinality \Group{G} \geq (\dim \HilbertSpace{})^2$ with the Haar measure $\Measure{g}$ on it. Then $\Operator\in\BoundedOperators{\HilbertSpace{}}$ is a separable operator iff there exist absolutely and square integrable mappings $\Function^\mu:\Group{G}\to\HilbertSpace{\mu}$ such that 
\begin{equation*}
\Operator = \IntegralMappingExpanded
{\Group{G}} 
{\ProjectorExpanded{\MultipleTensorConvolution{\Function}{\Multiplicity}(g)}} 
{\Measure{g}}.
\end{equation*}
\end{theorem}

This theorem is an immediate consequence of Theorems \ref{thm:Convolutional} and Proposition \ref{prop:CRepresentability}, which is stated in section \ref{sec:chapter6:Representability}. 

\begin{remark}
\label{rem:Separability}
Consider a Hilbert tensor  product $\HilbertSpace{}=\HilbertTensorProduct{\Multiplicity}$.
Let $\Group{G}$ be a locally compact abelian group of cardinality $\Cardinality \Group{G} \geq (\dim \HilbertSpace{})^2$ with the Haar measure $\Measure{g}$ on it.
By virtue of the above theorem the separability problem for $\HilbertSpace{}$ can be written as follows:
\begin{quote}
Determine whether for a given operator $\Operator\in\BoundedOperators{\HilbertSpace{}}$  there exist absolutely and square integrable mappings $\Function^\mu:\Group{G}\to\HilbertSpace{\mu}$ such that 
\begin{equation*}
\Operator = \IntegralMappingExpanded
{\Group{G}} 
{\ProjectorExpanded{\MultipleTensorConvolution{\Function}{\Multiplicity}(g)}} 
{\Measure{g}}.
\end{equation*}
\end{quote}
\end{remark}

In section \ref{sec:spectral} we formulate a problem of finding spectral decompositions of separable operators, and give two examples of nontrivial decompositions basing on the main result of this article.

\section{Preliminaries}

\subsection{Separable operators}

Let $\Multiplicity\geq 2$ be a natural number. For
$\mu = 1,\ldots,\Multiplicity$, let $\HilbertSpace{\mu}$ be a
finite $N_\mu$-dimensional Hilbert space with a Hermitian product
$\HProduct{\cdot}{\cdot}{\mu}$ (and according norm $|\cdot|_\mu$) $\CC$-linear with respect to the
second argument. Let $\HilbertSpace{} =
\HilbertTensorProduct{\Multiplicity}$ be the Hilbert tensor product with
Hermitian product $\HProduct{\cdot}{\cdot}{}$ given by the linear
extension of
$$\HProduct{\TensorProduct{v}{\Multiplicity}}{\TensorProduct{w}{\Multiplicity}}{} = \HProduct{v^1}{w^1}{1} \cdots \HProduct{v^\Multiplicity}{w^\Multiplicity}{\Multiplicity}$$
for $\TensorProduct{v}{\Multiplicity}$ and
$\TensorProduct{w}{\Multiplicity}$ in $\HilbertSpace{}$.

For each vetor $\Ket{v}\in\HilbertSpace{}$ denote by $\Bra{v}\in\Conjugate{\HilbertSpace{}}$ the functional acting on $\HilbertSpace{}$ associated with $\Ket{v}$ by the Hermitian product $\HProduct{\cdot}{\cdot}{}$ in the standard way. Denote by 
$\Projector{v}:\HilbertSpace{}\to\HilbertSpace{}$ the unnormalized Hermitian
projector operator (for abbreviation later called projector
operator or projector) on $\Ket{v}\in\HilbertSpace{}$, that is
\begin{eqnarray*}
\Projector{v}\Ket{w} = \HProduct{v}{w}{}\Ket{v}.
\end{eqnarray*}
Note that in the first two sections we were using the Dirac notation $\ProjectorExpanded{v}$ for such operators.
We denote by $\HermitianOperators{\HilbertSpace{}}$ the real linear space of
Hermitian operators acting on $\HilbertSpace{}$. Simple tensors ${v} =
{\TensorProduct{v}{\Multiplicity}}\in\HilbertSpace{}$  are called
separable and the projector $\Projector{\TensorProduct{v}{\Multiplicity}}$ on a separable vector is called a separable projector. Vectors that are not separable are called non-separable or entangled.

In $\HermitianOperators{\HilbertSpace{}}$ we distinguish the cone
$\SeparableOperators{\HilbertSpace{}}$ generated by the positive
combinations of separable projectors. If
$\Operator\in\SeparableOperators{\HilbertSpace{}}$ then we call it
separable.  Operators that are not separable are called non-separable or entangled. The following proposition is a consequence of \cite[Lemma 2]{zyczkowski1} or \cite[Corollary 1]{gurvits2}.

\begin{proposition}
\label{prop:Interior}
The set of separable operators
$\SeparableOperators{\HilbertSpace{}}$ is a closed convex cone
with nonempty interior in $\HermitianOperators{\HilbertSpace{}}.$
\end{proposition}

Denote by $\HermitianOperatorsDual{\HilbertSpace{}}$ the dual space to $\HermitianOperators{\HilbertSpace{}}$ of linear functionals acting on $\HermitianOperators{\HilbertSpace{}}$. For every closed convex cone $\Cone \subset\HermitianOperators{\HilbertSpace{}}$ the dual cone $\DualCone{\Cone}\subset\HermitianOperatorsDual{\HilbertSpace{}}$ is defined by
\begin{eqnarray*}
\DualCone{\Cone} = \{ \DualOperator \in \HermitianOperatorsDual{\HilbertSpace{}} \ | \ \forall_{\Operator\in\Cone}\ \DualOperator(\Operator)\geq 0 \}.
\end{eqnarray*}
By the Hahn-Banach theorem, the second dual cone equals the initial one, that is $\DualCone{\DualCone{\Cone}} = \Cone$. Therefore $\Operator\in\Cone$ iff $\DualOperator(\Operator) \geq 0$ for all $\DualOperator\in \DualCone{\Cone}$. 

Taking $\Cone = \SeparableOperators{\HilbertSpace{}}$ we get the following proposition. 

\begin{proposition}
\label{prop:SeparableCone}
A Hermitian operator $\Operator$ is separable iff $\DualOperator(\Operator) \geq 0$  for all $\DualOperator\in\DualCone{\SOHS}$.
\end{proposition}

Note that each trace-one operator $\DualOperator\in\DualCone{\SOHS}$ which is not positive semi-definite is called entanglement witness \cite{horodecki}. 

\subsection{\LCA~groups}
\label{sec:Fourier}

In this section we gather some useful facts on Fourier analysis on locally compact abelian groups (we refer to \cite{rudin} for a more comprehensive view). Let $\Group{G}$ be an LCA group. These are, for example, all abelian groups (e.g. $\ZZ_n$, $\ZZ^d$, $\QQ^d$, $\RR^d$, $\QuotientGroup{\RR^d}{\Lattice}$, where $\Lattice\subset\RR^d$ is a discrete subgroup) with discrete topology as well as $\RR^n$ and $\TT^d = \QuotientGroup{\RR^d}{\ZZ^d}$ with standard Euclidean topology ($\QQ^d$ with standard topology is not an \LCA~group since it is not locally compact). 
Consider all continuous homomorphisms $\Action{}$ from $\Group{G}$ to the unit circle in $\CC$. Denote by $\DualGroup{G}$ the set of such homomorphisms; we let an element $\Character\in\DualGroup{G}$ correspond to the homomorphism denoted by $\Action{\Character}:\Group{G}\to\CC$. Then $\DualGroup{G}$ is an abelian group with the neutral element $\NeutralElementDual$ s.t. $\Action{\NeutralElementDual} \equiv 1$ and the operation $"+"$ s.t. $\Action{\Character+\Character '}(g) = \Action{\Character}(g)\cdot\Action{\Character '}(g)$ for all $g\in\Group{G}$.
Now, $\DualGroup{G}$ embeded with the weak topology (the weakest topology for which $\Action{\cdot}(g):\DualGroup{G}\to\CC$ is continuous for all $g\in\Group{G}$) is an \LCA~group. Moreover, by the Pontryagin duality the double dual group $\DualGroup{\DualGroup{G}}$ is canonically isomorphic with $\Group{G}$. Some examples are in order. The groups $\ZZ_n$ and $\RR^d$ (with Euclidean topology) are self-dual, i.e. $\DualGroup{\ZZ}_n = \ZZ_n$ and ${\RR^d}^* = \RR^d$. The action of the element $\Character\in\DualGroup \ZZ_n$ is $\Action{\Character}(g) = e^{2\pi\II\cdot\Character g/n}$, and the action of the element $\Character\in{\RR^d}^*$ is $\Action{\Character} = e^{2\pi\II\cdot\ScalarProduct{\Character}{g}{}}$, where $\ScalarProduct{\cdot}{\cdot}{}$ is any scalar product. The dual group of $\ZZ^d$ is $\TT^d$ (with topology induced from the Euclidean one) and vice versa. The action of $\Character\in\TT^d$ on $\ZZ^d$ (as well as the action of $\Character\in\ZZ^d$ on $\TT^d$) is given by the same formula as in the $\RR^d$ case.

For each \LCA~group $\Group{G}$ there exists the Haar measure, i.e. the only (up to a positive multiplicative constant) regular non-negative Borel measure invariant with respect to all translations. We denote such measure by $\Measure{g}$. Now, for measurable $f,f':\Group{G}\to\CC$ such that $\IntegralMappingExpanded{\Group{G}}{|f(g-h)\cdot f'(h)|}{\Measure{h}} < \infty$ we define their convolution $\Convolution{f}{f'}:\Group{G}\to\CC$ in the standard way: $\Convolution{f}{f'}(g) = \IntegralMappingExpanded{\Group{G}}{f(g-h)\cdot f'(h)}{\Measure{h}}$. For our purposes it is important that for $p=1,2$, if $f\in\LL{1}(\Group,\CC)$ (i.e. it is absolutely integrable) and $f'\in\LL{p}(\Group{G},\CC)$ (i.e. $\IntegralMappingExpanded{\Group{G}}{|f'(g)|^p}{\Measure{g}} < \infty$), then $\Lnorm{\Convolution{f}{f'}}{p} \leq \Lnorm{f}{1}\cdot\Lnorm{f'}{p}$. Hence $\LL{1}(\Group{G},\CC)\cap\LL{2}(\Group{G},\CC)$ is closed under the convolution operation.

Given an \LCA~group $\Group{G}$, the Fourier transform $\FourierTransformMapping$ is the mapping from $\LL{2}(\Group{G},\CC)$ to $\LL{2}(\DualGroup{G},\CC)$ (recall that $\DualGroup{G}$ is an \LCA~group so it has the well defined Haar measure). Formally, we first define $\FourierTransform{f}:\DualGroup{G}\to\CC$ for each $f\in\LL{1}(\Group{G},\CC)$ by 
\begin{equation*}
\FourierTransform{f}(\Character) = \IntegralMappingExpanded
{\Group{G}}
{\Action{\Character}(-g)\cdot f(g)}
{\Measure{g}}.
\end{equation*} 
Then we distinguish a dense class of functions $f$ in $\LL{1}(\Group{G},\CC)\cap\LL{2}(\Group{G},\CC)$ for which $\FourierTransform{f}\in\LL{1}(\DualGroup{G},\CC)$ and finally extend the Fourier transform to all functions in $\LL{2}(\Group{G},\CC)$ with the image $\LL{2}(\DualGroup{G},\CC)$. The point is that for the Haar mesure $\Measure{g}$ on $\Group{G}$ there exists the Haar measure $\Measure{\Character}$ on $\DualGroup{G}$ such that for all $f, f'\in\LL{2}(\Group{G},\CC)$ the Parseval equality
\begin{equation}
\label{eqn:ParsevalEquality}
\IntegralMappingExpanded
{\Group{G}}
{\ComplexConjugate{f(g)}\cdot f'(g)}
{\Measure{g}} =
\IntegralMappingExpanded
{\DualGroup{G}}
{\ComplexConjugate{\FourierTransform{f}(\Character)}\cdot \FourierTransform{f'}(\Character)}
{\Measure{\Character}}
\end{equation}
holds true. Moreover, $\FourierTransformMapping$ is an isomorphism of the Hilbert spaces, so there exists the inverse Fourier transform $\InverseFourierTransformMapping:\LL{2}(\DualGroup{G},\CC) \to\LL{2}(\Group{G},\CC)$ given by
\begin{equation}
\label{eqn:InvereseFourierTransform}
\InverseFourierTransform{\hat f}(g) = \IntegralMappingExpanded
{\DualGroup{G}}
{\Action{\Character}(g)\cdot \hat f(\Character)}
{\Measure{\Character}}
\end{equation} 
for $\hat f:\DualGroup{G}\to\CC$.

Finally, we need to know that the convolution of functions on $\Group{G}$ Fourier-transforms to the multiplication of the transformed functions, i.e.
\begin{equation}
\label{eqn:Convolution}
\FourierTransform{\Convolution f {f'}}(\Character) = \FourierTransform{f}(\Character)\cdot\FourierTransform{f'}(\Character)
\end{equation}
for all $f,f'\in\LL{1}(\Group{G},\CC)\cap\LL{2}(\Group{G},\CC)$ and $\Character\in\DualGroup{G}$.

\section{Proof of Theorem \ref{thm:Convolutional}}
\label{sec:Proof1}

Recall that we want to prove that if $\Group{G}$ is an \LCA~group with the Haar measure $\Measure{g}$, and if $\Function=(\Function^1,\ldots,\Function^\Multiplicity)\in\ReferenceSetConvolutional$,  then the Hermitian operator 
\begin{equation*}
\tag{\ding{73}}
\OperatorRepresentedBy{\Function} = \IntegralMappingExpanded
{\Group{G}} 
{\Projector{\MultipleTensorConvolution{\Function}{\Multiplicity}(g)}} 
{\Measure{g}},
\end{equation*}
acting on $\HilbertSpace{}=\HilbertTensorProduct{\Multiplicity}$, is separable.

First of all, notice that the convolution (tensor convolution) is naturally considered on the space of integrable functions (mappings), but we will also use the Fourier transform and the Parseval equality which is true for square integrable functions (mappings). Hence we choose above space of mappings because of technical reasons. Namely, as we mentioned $\LL{1}(\Group{G},\CC)\cap\LL{2}(\Group{G},\CC)$ is closed under the convolution operation. Hence for $\Function^\mu\in\LL{1}(\Group{G},\HilbertSpace{\mu}) \cap\LL{2}(\Group{G},\HilbertSpace{\mu})$, $\TensorConvolution{\Function^1}{\Function^2}\in \LL{1}(\Group{G},\HilbertSpace{1}\otimes\HilbertSpace{2}) \cap\LL{2}(\Group{G},\HilbertSpace{1}\otimes\HilbertSpace{2})$ and so on. Finally, we get that 
\begin{equation*}
\MultipleTensorConvolution{\Function}{\Multiplicity} \in\LL{1}(\Group{G},\HilbertTensorProduct{\Multiplicity}) \cap\LL{2}(\Group{G},\HilbertTensorProduct{\Multiplicity}).
\end{equation*}
Therefore, the integral in (\ref{eqn:CRepresentation}) is well defined. 

\begin{proof}[Proof of Proposition \ref{prop:21}]
Since separable vectors in $\HilbertSpace{}$ span all the space, a mapping $\Function$ in $\SquareIntegrableFunctions{\Group{G}}{\HilbertSpace{}}$ is unambiguously defined by functions
\begin{equation}
\label{eqn:32}
\Function[v] := \HermitianProduct{\TensorProduct{v}{\Multiplicity}}{\Function(\cdot)}{} :\Group{G}\to \CC,
\end{equation} 
where $v=\TensorProduct{v}{\Multiplicity}\in\SOHS$ is arbitrary. If $\Function$ is of the form $\MultipleTensorConvolution{\Function}{\Multiplicity}$, then $\Function[v]$ is, in fact, a standard convolution of functions $\Function^\mu[v^\mu] := \HermitianProduct{v^\mu}{\Function^\mu(\cdot)}{\mu}:\Group{G}\to\CC$, that is
\begin{equation}
\label{eqn:33}
\Function[v] = \Convolution{\Convolution{\Function^1[v^1]}{\cdots}} {\Function^\Multiplicity[v^\Multiplicity]}.
\end{equation} 
Therefore, using (\ref{eqn:Convolution}) we get that
\begin{equation}
\label{eqn:FourierTrnasformConvolutional}
\begin{split}
\FourierTransform{\Function[v]}(\Character) &= \FourierTransform{\Function^1[v^1]}(\Character)\cdots \FourierTransform{\Function^{\Multiplicity}[v^{\Multiplicity}]}(\Character) \\
&= 
\HermitianProduct{v^1}{\FourierTransform{\Function^1}(\Character)}{1} \cdots \HermitianProduct{v^\Multiplicity} {\FourierTransform{\Function^\Multiplicity}(\Character)}{\Multiplicity} \\
&= 
\HermitianProduct{\TensorProduct{v}{\Multiplicity}} {\FourierTransform{\Function^1}(\Character)\otimes \cdots\otimes\FourierTransform{\Function^{\Multiplicity}}(\Character)}{},
\quad \Character\in\DualGroup{G}.
\end{split}
\end{equation}
In the second line, we use the fact that $\FourierTransformMapping$ is a linear operator.
Consequently, $$\FourierTransform{ \MultipleTensorConvolution{\Function}{\Multiplicity}}(\Character) = \FourierTransform{\Function^1}(\Character)\otimes \cdots\otimes\FourierTransform{\Function^{\Multiplicity}}(\Character).$$ 

Now, using the formula \eqref{eqn:CRepresentation} for $\OperatorRepresentedBy{\Function}$, we compute for a given vector $v = \TensorProduct{v}{\Multiplicity}$,
\begin{align*}
\HermitianProduct{v}{\OperatorRepresentedBy{\Function}|v}{} 
&= 
\IntegralMappingExpanded
{\Group{G}} 
{\HermitianNorm{\MultipleTensorConvolution{\Function}{\Multiplicity}[v](g)}{}^2} 
{\Measure{g}} \\
&= 
\IntegralMappingExpanded
{\Group{G}} 
{\HermitianNorm{\Convolution{\Convolution{\Function^1[v^1]}{\cdots}} {\Function^\Multiplicity[v^\Multiplicity](g)}
}{}^2} 
{\Measure{g}} & \text{by \eqref{eqn:33}} \\
&= 
\IntegralMappingExpanded
{\DualGroup{G}} 
{\HermitianNorm{\FourierTransform{\Function^1[v^1]}(\Character)\cdots \FourierTransform{\Function^\Multiplicity[v^\Multiplicity]}(\Character)}{}^2} 
{\Measure{\Character}} & \text{by Parseval's equality} \\
&= 
\IntegralMappingExpanded
{\DualGroup{G}} 
{\HermitianNorm{\left(\FourierTransform{\Function^1}(\Character)\otimes \cdots\otimes\FourierTransform{\Function^{\Multiplicity}}(\Character)\right)[v]}{}^2} 
{\Measure{\Character}} & \text{by \eqref{eqn:FourierTrnasformConvolutional} and \eqref{eqn:32}} \\
&= 
\HermitianProduct{v}
{
\IntegralMappingExpanded
{\DualGroup{G}} 
{\Projector{\FourierTransform{\Function^1}(\Character)\otimes \cdots\otimes\FourierTransform{\Function^{\Multiplicity}}(\Character)}} 
{\Measure{\Character}}
\mid v}{},
\end{align*}
where $\Measure{\Character}$ is the Haar measure on $\DualGroup{G}$ conjugated with $\Measure{g}$, and Parseval's equality is given by \eqref{eqn:ParsevalEquality}. Therefore, 
\begin{eqnarray*}
\OperatorRepresentedBy{\Function} = 
\IntegralMappingExpanded
{\DualGroup{G}} 
{\Projector{\FourierTransform{\Function^1}(\Character)\otimes \cdots\otimes\FourierTransform{\Function^{\Multiplicity}}(\Character)}} 
{\Measure{\Character}}.
\end{eqnarray*}
\end{proof}

\begin{proof}[Proof of Theorem \ref{thm:Convolutional}]
We use the fact that the set of separable operators $\SOHS$ is a closed convex cone, and therefore, by Proposition \ref{prop:SeparableCone}, $\Operator\in\SOHS$ iff $\DualOperator(\Operator) \geq 0$ for all $\DualOperator\in\DualCone{\SOHS}$. 

Let us take absolutely and square integrable $\Function^\mu:\Group{G}\to\HilbertSpace{\mu}$ and assume $\OperatorRepresentedBy{\Function}$ is given by (\ref{eqn:CRepresentation}). By Proposition \ref{prop:21} it is an integral of separable operators, and so for every $\DualOperator\in\DualCone{\SOHS}$, $\DualOperator(\OperatorRepresentedBy{\Function})$ can be seen as an integral of a non-negative function. Therefore, $\DualOperator(\OperatorRepresentedBy{\Function})\geq 0$, so we end the proof using Proposition \ref{prop:SeparableCone}.
\end{proof}

Note that Proposition \ref{prop:SeparableCone} was used only for a technical purpose. Since $\SOHS$ is a closed convex cone and 
$\Projector{\FourierTransform{\Function^1}(\Character)\otimes \cdots\otimes\FourierTransform{\Function^{\Multiplicity}}(\Character)}\in\SOHS$, it is intuitive that their "infinite positive combination" is separable as well. It is because an integral of a mapping with values in a closed convex cone has value in this cone, which is a consequence of the Hahn-Banach separation theorem.

\section{Proof of Theorem \ref{thm:3}}
\label{sec:chapter6:Representability}

Theorem \ref{thm:3} will follow immediately from Theorem \ref{thm:Convolutional} and underneath Proposition \ref{prop:CRepresentability}. Before we state it, we need one definition. A separable operator $\Operator$ is called \emph{$P$-separable} if there exist $P\in\NN$ and  separable vectors $v_p\in\HilbertSpace{}$ (not necessarily different from $0\in\HilbertSpace{}$) for $p\TruncatedNumbers{P}$ such that
$$
\Operator = \sum_{p=1}^P \Projector{v_p}.
$$
Clearly, if $Q\leq P$, then $Q$-separability of an operator implies its $P$-separability.

\begin{proposition}
\label{prop:CRepresentability}
Let $\Group{G}$ be an \LCA~group. If the cardinality of $\Group{G}$ is not less then $P$, that is $\Cardinality{\Group{G}} \geq P$, then every $P$-separable operator on $\HilbertSpace{}$ is representable in the form \eqref{eqn:CRepresentation} in Theorem \ref{thm:Convolutional}.
\end{proposition}

The proof will immediately follow from the ensuing lemma.

\begin{lemma}
\label{lem:321}
Let $\Group{G}$ be an \LCA~group and $\DualGroup{G}$ its dual. For a given open set $U\subset\DualGroup{G}$ there exists a continuous function $f:{\DualGroup{G}}\to{\CC}$ with compact support contained in $U$ such that $\Lnorm{f}{2\Multiplicity} > 0$ and $\InverseFourierTransform{f}\in\IntegrableFunctions{\Group{G}}{\CC}$ ($\InverseFourierTransformMapping$ is the inverse Fourier transform given by (\ref{eqn:InvereseFourierTransform})).
\end{lemma}

\begin{proof}
Without loss of generality, we can consider only the case $U$ is an open neighbourhood of the neutral element $\eta$ in $\DualGroup{G}$ since the topology of $\DualGroup{G}$ is uniform with respect to the group operation (the shifts).

Take an open neighbourhood $U'$ with the compact closure contained in $U$.
For such neighbourhood there exists an open neighbourhood $V$ of $\eta$ such that $V= -V$ and $V+V\subset U'$. Indeed, since $\eta+\eta = \eta$ and by the continuity of the group operation there exist $V_1$,$V_2$ - open neighbourhoods of $\eta$ such that $V_1+V_2\subset U'$. Hence putting $V=V_1\cap V_2 \cap (-V_1)\cap (-V_2)$ we obtain the assertion. Note that the closure of $V + V$ is compact and is contained in $U$.

Consider the indicator function $\IndicatorFunction{V}$ of the set $V$. Put $f=\Convolution{\IndicatorFunction{V}}{\IndicatorFunction{V}}$. By the definition,  $f$ is strictly positive on the open set $V+V$ and $\Supp f = \Cl (V + V)$, hence $\Supp f$ is compact and is contained in $U$. Moreover, on every \LCA~group open sets have strictly positive measures \cite[Section 1.1.2]{rudin}. Therefore, by the continuity of $f$, $\Lnorm{f}{2\Multiplicity} > 0$.

Finally, $f\in\IntegrableFunctions{\DualGroup{G},\CC}$ as well. It is known that $f$ is a positive definite function \cite[Section 1.4.2]{rudin}, so by the Bochner \cite[Section 1.4.3]{rudin} and the inversion \cite[Section 1.5.1]{rudin} theorems the inverse Fourier transform $\InverseFourierTransform{f}$ is integrable.
\end{proof}

\begin{lemma}
\label{lem:322}
Let $\Group{G}$ be an \LCA~group and $\DualGroup{G}$ its dual. If for a given $P\in\NN$ there exist mutually disjoint open sets $\Set{U_p\subset\DualGroup{G}}{p\TruncatedNumbers{P}}$, then every $P$-separable operator on $\HilbertSpace{}$ is representable in the form \eqref{eqn:CRepresentation}.
\end{lemma}

\begin{proof}
Assume we want to represent a separable operator $\sum_{p=1}^{P}\Projector{\TensorProduct{v_p}{\Multiplicity}}$. 
Take continuous functions $f_p:{\DualGroup{G}}\to{\CC}$ with compact supports contained in $U_p$ such that $\Lnorm{f_p}{2\Multiplicity} > 0$ and $\InverseFourierTransform{f_p}\in\IntegrableFunctions{\Group{G}}{\CC}$, respectively. Such functions exist by Lemma \ref{lem:321}. Then $f_p\in\LL{q}({\DualGroup{G}},{\CC})$ for $q=1,2,2\Multiplicity$ (actually for  all $1\leq q\leq \infty$). 
Define $\Psi^\mu:\DualGroup{G}\to\HilbertSpace{\mu}$ by 
\begin{align*}
\Psi^\mu(\Character) &= \sum_p \frac{v^\mu_p}{\Lnorm{f_p}{2\Multiplicity}} \cdot f_p(\Character), & \mu\TruncatedNumbers{\Multiplicity}.
\end{align*}
Put $\Function^\mu = \InverseFourierTransform{\Psi^\mu}$. Then $(\Function^1,\ldots,\Function^\Multiplicity)$ is in $\ReferenceSetConvolutional$ since the Fourier transform is an isomorphism of the spaces of square integrable functions and we assumed that $\InverseFourierTransform{f_p}\in\LL{1}(\Group{G})$. By the Parseval equality  (\ref{eqn:ParsevalEquality}) and formula \eqref{eqn:36},
\begin{align*}
\OperatorRepresentedBy{\Function} &=
\sum_p
\frac{1}{\Lnorm{f_p}{2\Multiplicity}^{2\Multiplicity}} 
\IntegralMappingExpanded
{\DualGroup{G}}
{\Projector{\TensorProduct{v_p}{\Multiplicity}}\cdot |f_p(\Character)}|^{2\Multiplicity}
{\Measure{\Character}} \\
&=
\sum_p
{\Projector{\TensorProduct{v_p}{\Multiplicity}}},
\end{align*} 
since $U_p$ are mutually disjoint.
\end{proof}

\begin{proof}[Proof of Proposition \ref{prop:CRepresentability}]
Assume that the cardinality of $\Group{G}$ is at least $P$. Then the same is true for its dual $\DualGroup{G}$. Indeed, since \LCA~groups are, by the definition, Hausdorff the only possible topology on finite groups is the discrete one (it follows, for example, from the fact that every set composed of only one element is the intersection of all closed neighbourhoods of the element). 
It is well known that every finite abelian group $\DiscreteGroup$ is of the form $\ZZ_{n_1}\times\cdots\times\ZZ_{n_k}$, which implies that they are self-dual (with respect to the discrete topology structure), that is $\DiscreteGroup\simeq\DualGroup{\DiscreteGroup}.$ Now, from the duality between compact and discrete groups, we conclude that the only groups $\Group{G}$ with  $\Cardinality{\DualGroup{G}}<P$ are finite of cardinality less than $P$.

Since $\Cardinality{\DualGroup{G}}\geq P$ and $\DualGroup{G}$ is Hausdorff, there exist $P$ mutually disjoint open sets $\{ U_p \}$ in it. Hence the proof is complete by Lemma \ref{lem:322}.
\end{proof}

\begin{remark}
In Theorem \ref{thm:3} the spaces $\IntegrableFunctions{\Group{G}}{\HilbertSpace{\mu}} \cap \SquareIntegrableFunctions{\Group{G}}{\HilbertSpace{\mu}}$ can be replaced by smaller ones. Namely, in the proof of Proposition \ref{prop:CRepresentability} we use, in fact, mappings $\Function^\mu:\Group{G}\to\CC$ with compactly supported continuous Fourier transforms $\FourierTransform{\Function^\mu}$. 
\end{remark}

\begin{proof}[Proof of Theorem \ref{thm:3}]
By Caratheodory's theorem in convex analysis \cite[{Theorem 17.1}]{rockafellar} all separable operators are $(\dim \HilbertSpace{})^2$-separable. Therefore, by Proposition \ref{prop:CRepresentability}, for each separable operator $\Operator\in\SOHS$ there exist $\Function\in \ReferenceSetConvolutional$ such that $\Operator = \OperatorRepresentedBy{\Function}$ given by (\ref{eqn:CRepresentation}). The converse is a consequence of Theorem \ref{thm:Convolutional}.
\end{proof}

\section{Spectral decompositions of separable operators}
\label{sec:spectral}

Since separable operators are only positive combinations of separable projectors one can not expect that the spectral decomposition of each separable operator consists of separable projectors. Thus it is advisable to consider a problem of characterizing spectral decompositions of separable operators. In the framework of our considerations one could consider the following problem.

\begin{quotation}
For a given Hilbert space $\HilbertSpace{}=\HilbertTensorProduct{\Multiplicity}$, find all discrete groups $\Group{G}$ and all tuples of mappings $\SetSuchThat{v^\mu:\Group{G}\to\HilbertSpace{\mu}}{\mu\TruncatedNumbers{\Multiplicity}}$ such that  for all $g,h\in\Group{G}$ one of the following conditions is satisfied:
\renewcommand{\theenumi}{(\Alph{enumi})}
\renewcommand{\labelenumi}{\theenumi}
\begin{enumerate}
\item \label{enum:a} $\HermitianProduct{\MultipleTensorConvolution{v}{\Multiplicity}(g)} {\MultipleTensorConvolution{v}{\Multiplicity}(h)}{} = 0$, 
\item \label{enum:b} there exists $\lambda\in\RR$ such that $\MultipleTensorConvolution{v}{\Multiplicity}(g) = \lambda \MultipleTensorConvolution{v}{\Multiplicity}(h).$
\end{enumerate}
\renewcommand{\theenumi}{\arabic{enumi}}
\renewcommand{\labelenumi}{\theenumi}
\end{quotation}

In this case $\sum_{g\in\Group{G}}{\Projector{\MultipleTensorConvolution{v}{\Multiplicity}(g)}}$
is the spectral decomposition (up to normalization of each of the projectors) of a Hermitian operator and by Theorem \ref{thm:Convolutional} this operator is separable.

Let us consider two examples. The easiest one is when $\HilbertSpace{} = \HilbertSpace{1}\otimes\HilbertSpace{2}$ and $\Group{G} = \ZZ_2$ with counting measure. In this setting we have the following proposition.

\begin{proposition}
\label{prop:SpectralDecompZZ_2}
Suppose $v^1:\ZZ_2\to\HilbertSpace{1}$ and $v^2:\ZZ_2\to\HilbertSpace{2}$ are mappings with values of the same norm. Then $\Projector{\TensorConvolution{v^1}{v^2}(0)} + \Projector{\TensorConvolution{v^1}{v^2}(1)}$ is the spectral decomposition (up to  normalization of each of the projectors) of a separable operator acting on $\HilbertSpace{1}\otimes\HilbertSpace{2}$ iff
one of the following conditions is satisfied:
\renewcommand{\theenumi}{(\alph{enumi})}
\renewcommand{\labelenumi}{\theenumi}
\begin{enumerate}
\item \label{enum:a1} $\Angle{v^1_0}{v^1_1}{1} \pm \Angle{v^2_0}{v^2_1}{2} = \pi \Mod 2\pi$;
\item \label{enum:b1} $v^\mu_0, v^\mu_1$ are linearly dependent for $\mu=1,2$,
\end{enumerate}
\renewcommand{\theenumi}{(\arabic{enumi})}
\renewcommand{\labelenumi}{\theenumi} 
where $\Angle{v^\mu_0}{v^\mu_1}{\mu}$ denotes the angle between $v^\mu_0$ and $v^\mu_1$ with respect to the scalar product $\ScalarProduct{\cdot}{\cdot}{\mu}$.
\end{proposition}

Note that we put arguments of the functions $v^\mu$ in the subscripts and leave standard notation in case of their convolution $\TensorConvolution{v^1}{v^2}$.

\begin{proof}
By Theorem \ref{thm:Convolutional}, $\Projector{\TensorConvolution{v^1}{v^2}(0)} + \Projector{\TensorConvolution{v^1}{v^2}(1)}$ is a separable operator. 
We check conditions \ref{enum:a}, \ref{enum:b} in the only nontrivial (and nonequivalent) case $g=0$ and $h=1$. We will show that \ref{enum:a}, \ref{enum:a1} and \ref{enum:b}, \ref{enum:b1} are pairwise equivalent.

Equivalence of \ref{enum:a} and \ref{enum:a1}. It follows that
\begin{align*}
\HermitianProduct{\TensorConvolution{v^1}{v^2}(0)} {\TensorConvolution{v^1}{v^2}(1)}{} 
&=
\HermitianProduct{v^1_0}{v^1_1}{1}\HermitianNorm{v^2_0}{2}^2
+ \HermitianProduct{v^1_1}{v^1_0}{1}\HermitianNorm{v^2_1}{2}^2 \\
&\quad + \HermitianNorm{v^1_0}{1}^2\HermitianProduct{v^2_0}{v^2_1}{2}
+ \HermitianNorm{v^1_1}{1}^2\HermitianProduct{v^2_1}{v^2_0}{2}.
\end{align*}
Since all the vectors have the same norm $R$ we get
\begin{eqnarray*}
\HermitianProduct{\TensorConvolution{v^1}{v^2}(0)} {\TensorConvolution{v^1}{v^2}(1)}{} = 2R^2[\cos\Angle{v^1_0}{v^1_1}{1} + \cos\Angle{v^2_0}{v^2_1}{2}],
\end{eqnarray*}
where $
\cos\Angle{v}{v'}{\mu} = {\ScalarProduct{v}{v'}{\mu}}/({\HermitianNorm{v}{\mu}\cdot\HermitianNorm{v'}{\mu}}).$   From this formula \ref{enum:a1} follows easily.

Equivalence of \ref{enum:b} and \ref{enum:b1}.
If $v^1_0\otimes v^2_0 + v^1_1\otimes v^2_1 = \lambda\cdot( v^1_0\otimes v^2_1 + v^1_1\otimes v^2_0)$, then $v^1_0\otimes(v^2_0 - \lambda\cdot v^2_1) = v^1_1\otimes(\lambda\cdot v^2_0 - v^2_1)$. Hence, $v^1_0$ and $v^1_1$ are linearly dependent. Similarly, we prove linear dependence of $v^2_0$ and $v^2_1$. The condition \ref{enum:b1} is therefore fulfilled. 
\end{proof}

Note that in case condition \ref{enum:b} is satisfied, we obtain a trivial solution, i.e. $\Projector{\TensorConvolution{v^1}{v^2}(0)} + \Projector{\TensorConvolution{v^1}{v^2}(1)}$ is proportional to $\Projector{v^1_0\otimes v^2_0}.$

In the second example we consider the  group $\ZZ_n$ with counting measure, and
 $\HilbertSpace{} = \HilbertTensorProduct{\Multiplicity}$ with $\dim\HilbertSpace{\mu} = N_\mu\geq n$. 

\begin{proposition}
\label{prop:SpectralDecompZZ_n}
For each $\mu\TruncatedNumbers{\Multiplicity}$, let $\SetSuchThat{\EE{\mu}_g}{g\in\ZZ_n}$ be an orthonormal system in $\HilbertSpace{\mu}$, and $\lambda^\mu:\ZZ_n\to\CC$, $g\mapsto \lambda_g^\mu\in\CC$ be a function. Consider $v^\mu:\ZZ_n\to\HilbertSpace{\mu}$, $g\mapsto v^\mu_g = \lambda_g^\mu\cdot\EE{\mu}_g$. Then the folowing holds.
\begin{enumerate}[(i)]
\item $\sum_{g\in\ZZ_n} \Projector{\MultipleTensorConvolution{v}{\Multiplicity}(g)}$ 
is the spectral decomposition (up to normalization of each of the projectors) of a separable operator on $\HilbertSpace{}.$
\item Additionaly, if $\Multiplicity \geq 2$, and $\lambda^\mu$ is a constant function for each $\mu\TruncatedNumbers{\Multiplicity}$, then for all $g\in\ZZ_n$ vector $\MultipleTensorConvolution{v}{\Multiplicity}(g)$ reduces to a homothety on each of the compound spaces $\HilbertSpace{\mu}$.
\end{enumerate}
\end{proposition}

\begin{proof}
(i) Separability of the operator is a consequence of Theorem \ref{thm:Convolutional}. It remains to prove that $\MultipleTensorConvolution{v}{\Multiplicity}(g)$, for $g\in\Group{G}$, are mutually orthogonal. With the use of expanded formula for the tensor convolution \eqref{eqn:TensorConvolution}, we get
\begin{eqnarray*}
\MultipleTensorConvolution{v}{\Multiplicity}(g) = \sum_{g^2,\ldots,g^\Multiplicity\in\ZZ_n}
\lambda^1_{g-g^2-\ldots -g^m}\lambda^2_{g^2}\cdots\lambda^\Multiplicity_{g^\Multiplicity}
\EE{1}_{g-g^2-\ldots -g^m}\otimes\EE{2}_{g^2}\otimes\cdots\otimes\EE{\Multiplicity}_{g^\Multiplicity}.
\end{eqnarray*}
Since $\Set{\EE{\mu}_g}$ is an orthonormal system, it follows that
\begin{multline*}
\HermitianProduct{\MultipleTensorConvolution{v}{\Multiplicity}(g)} {\MultipleTensorConvolution{v}{\Multiplicity}(h)}{} = \\
= \sum_{g^2,\ldots,g^\Multiplicity\in\ZZ_n}
\sum_{h^2,\ldots,h^\Multiplicity\in\ZZ_n}
\Big(
\ComplexConjugate{\lambda}^1_{g-g^2-\ldots -g^m}\ComplexConjugate{\lambda}^2_{g^2}\cdots \ComplexConjugate{\lambda}^\Multiplicity_{g^\Multiplicity}\cdot
\lambda^1_{h-h^2-\ldots -h^m}\lambda^2_{h^2}\cdots\lambda^\Multiplicity_{h^\Multiplicity} \\
\cdot\delta_{g-g^2-\ldots -g^m,h-h^2-\ldots -h^m}\delta_{g^2,h^2}\cdots\delta_{g^\Multiplicity,h^\Multiplicity}\Big),
\end{multline*}
where $\delta_{g,h}=\delta_g(h)$ is the Kronecker delta on $\ZZ_n$. Finally, we compute
\begin{eqnarray}
\label{eqn:2}
\HermitianProduct{\MultipleTensorConvolution{v}{\Multiplicity}(g)} {\MultipleTensorConvolution{v}{\Multiplicity}(h)}{} =
\prod_{\mu}\left(\sum_{g^\mu\in\ZZ_n}
|\lambda^\mu_{g^\mu}|^2 \right)\cdot\delta_{g,h},
\end{eqnarray}
which means that $\MultipleTensorConvolution{v}{\Multiplicity}(g)$ are mutually orthogonal.

(ii) Without loss of generality we can assume $\lambda^\mu \equiv 1$ for all $\mu\TruncatedNumbers{\Multiplicity}$. Then tracing out the operator $\Projector{\MultipleTensorConvolution{v}{\Multiplicity}(g)}$ with respect to all but the first factor of the Hilbert tensor product $\HilbertSpace{}$ we get the reduced operator
\begin{align*}
 & \sum_{g^2,\ldots,g^\Multiplicity\in\ZZ_n}
\sum_{h^2,\ldots,h^\Multiplicity\in\ZZ_n}
\SimpleOperatorExpanded{\EE{1}_{g-g^2-\ldots -g^m}}{\EE{1}_{g-h^2-\ldots -h^m}}
\cdot\delta_{g^2,h^2}\cdots\delta_{g^\Multiplicity,h^\Multiplicity} \\
& =
\sum_{g^2,\ldots,g^\Multiplicity\in\ZZ_n}
\ProjectorExpanded{\EE{1}_{g-g^2-\ldots -g^m}} \\
&=
\sum_{g^3,\ldots,g^\Multiplicity\in\ZZ_n}\Big(\sum_{g^2\in\ZZ_n}
\Projector{\EE{1}_{(g-g_3-\ldots -g^m)-g_2}}\Big) \\
&=
\sum_{g^3,\ldots,g^\Multiplicity\in\ZZ_n}\Identity{1} 
= n^{\Multiplicity - 2} \Identity{1},
\end{align*}
where $\Identity{1}$ denotes the identity operator on $\HilbertSpace{1}$. Similar calculation for $\mu = 2,\ldots,\Multiplicity$ gives the result.
\end{proof}

It is worth noticing that the operator in the above proposition is, in fact, the  projection operator (up to a positive multiplicative constant) on the subspace $$\Span\SetSuchThat{\MultipleTensorConvolution{v}{\Multiplicity}(g)} {g\in\ZZ_n}\subset\HilbertSpace{}.$$ It is because $\MultipleTensorConvolution{v}{\Multiplicity}(g)$ are mutually orthogonal and have the same norm, which is a consequence of \eqref{eqn:2}. Moreover, if we assume $\Multiplicity = 2$, then by (ii) we conclude that there exists a separable operator which spectral-decomposes in a basis of maximally entangled operators.

\section{Concluding remarks}

In this article we presented an approach to the quantum separability problem based on the theory of locally compact abelian groups. 
We showed that on each such group an integral with respect to the Haar measure of the convolution of Hilbert-space-valued mappings gives a separable operator. Then we proved that for large enough groups this observation gives rise to a reformulation of the separability problem, i.e. an operator is separable iff it is an integral of the convolution of appropriate mappings.

In the proofs of main theorems of the article we used two crucial facts. Namely, that there exist a unitary transform between two Hilbert spaces (Parseval's equality) and that there is an operation with values in the first space (convolution) which transforms appropriately with respect to the unitary transform. Therefore, it is a challenge to find different settings with such objects. For example, in \cite{jakubczyk} the authors gave an integral representation of separable states in a bipartite setting, where instead of convolution they considered a differential operator. Unfortunately, they did not get separability criterion as, for example, separable projectors could not be represented in the given framework. More extensive studies of this approach is a subject of \cite{pietrzkowski}.

One could also try to extend our result to the case of compact groups, on which Peter-Weyl theorem establishes a (non-canonical) orthonormal basis on $\SquareIntegrableFunctions{\Group{G}}{\CC}$. The problems are, however, that the dual object to $\Group{G}$, namely the set of irreducible unitary representations, does not posses a natural group structure and, what is even more problematic, the convolution "transforms" to a non-commutative product, hence is not appropriate to our purpose. In fact, a different approach to quantum separability problem of bipartite systems on compact groups has been proposed and studied by Korbicz, Lewenstein and Wehr \cite{korbicz1,korbicz2}, where they translate the property of separability from the algebra of operators on a bipartite Hilbert tensor product to the algebra of functions on a product of a compact group (with convolution as the operation). In particular, in this approach the group must not be abelian since each Hilbert space is identified with a representation space of its unitary representation.

The problem of finding possible spectral decompositions of separable operators posed in section \ref{sec:spectral} may provide a better understanding of the set of separable operators, thus giving new insights into the filed of entanglement science. It seams that approach presented in this article can be helpful in this investigations.

\section*{Acknowledgements}

This work was partially supported by the Polish Ministry of Research and Higher Education grant N201 039 32/2703, 2007-2010.

The author is grateful to Bronis\l aw Jakubczyk and Marek Ku\'s for helpful discussions.


\end{document}